\documentclass[11pt]{llncs}
\usepackage{fullpage}
\usepackage{amsmath}
\usepackage{amssymb}
\usepackage[pdftex]{graphicx}
\begin{document}
\title{Distributed Agreement in Tile Self-Assembly}

\author{Aaron Sterling\thanks{This research was supported in part by National Science Foundation Grants 0652569 and 0728806.}}

\institute{Department of Computer Science, Iowa State University.  \email{sterling@iastate.edu}}

\maketitle
\begin{abstract}
Laboratory investigations have shown that a formal theory of fault-tolerance will be essential to harness nanoscale self-assembly as a medium of computation.  Several researchers have voiced an intuition that self-assembly phenomena are related to the field of distributed computing. This paper formalizes some of that intuition.  We construct tile assembly systems that are able to simulate the solution of the wait-free consensus problem in some distributed systems.  (For potential future work, this may allow binding errors in tile assembly to be analyzed, and managed, with positive results in distributed computing, as a ``blockage'' in our tile assembly model is analogous to a crash failure in a distributed computing model.)  We also define a strengthening of the ``traditional'' consensus problem, to make explicit an expectation about consensus algorithms that is often implicit in distributed computing literature.  We show that solution of this strengthened consensus problem can be simulated by a two-dimensional tile assembly model only for two processes, whereas a three-dimensional tile assembly model can simulate its solution in a distributed system with any number of processes.
\end{abstract}
\section{Introduction}
One emerging field of computer science research is the algorithmic harnessing of molecular self-assembly to produce structures (and perform computations) at the nano scale.  In his Ph.D. thesis in 1998, Winfree~\cite{winfree} used tiles on the integer plane to define a self-assembly model which has become an influential tool.  As noted by several researchers (for example~\cite{report}), problems in algorithmic tile self-assembly share characteristics with better-studied problems in distributed computing: asynchronous computation, the importance of fault tolerance, and the limitations of local knowledge.  In this paper, we formalize a connection between the two fields, by constructing models of tile assembly that simulate solutions to the wait-free consensus problem in some distributed systems.

The tile self-assembly literature has considered two main classes of models: models in which tiles bind to one another in an error-free manner, and models in which there is a positive probability that mismatched tiles will bind to one another.  In an error-permitting model, if mismatched tiles bind, they can produce a \emph{blockage}---an unplanned tile configuration that stops a particular section of an assembly from being able to accrete tiles.

As blockages do occur in wetlab self-assembly experiments, it is natural to ask how we could make our self-assembly computations as resilient against blockages as possible.  Researchers have investigated mechanisms to limit the chance of blockages through error-correction (for example~\cite{chen}~\cite{fuji}), or, relatedly, for a tile assembly to ``heal'' itself in the event of damage~\cite{selfhealing}.  Like other error-correcting codes, these mechanisms can consume significant overhead, and only reduce without eliminating the possibility of a blockage.  Our interest in this paper, therefore, is to build a framework for robust self-assembly even in the presence of one or more unhealable blockages.  Of course, if we consider a situation in which multiple subassemblies grow independently, then the blockage of one subassembly will have no effect on the others.  The problem arises when otherwise independent subassemblies send information to, or receive information from, one another, and need to coordinate based on that information---hence our motivation to import the consensus problem into the world of tile assembly.

The most common types of processor faults modeled in distributed computing are \emph{crash failure} (where a processor stops functioning) and \emph{Byzantine failure} (where a processor can behave maliciously and take ``worst-possible'' steps).  Several other types of failure have been defined; in general, their severity lies between crash failure and Byzantine failure.  We will focus on shared objects that can be simulated in the presence of a tile self-assembly analogue of crash failures, to construct a theoretical foundation for synchronized fault-tolerance in self-assembly.  In the long run, we believe that the combination of error-correction and distributed computing techniques (to manage a variety of failures) will produce self-assembling systems with high fault-tolerance.

The consensus problem was originally defined by Lamport for a system of distributed processors, as an abstraction of the transaction commit problem in database theory.  It has since been shown to have wide application to the study of distributed systems; see, for example, Attiya and Welch~\cite{attiya} for a textbook introduction.  In brief, given a system of $n$ distributed processors, a \emph{solution to the consensus problem} is an algorithm that ensures all nonfaulty processors agree on the same value.  (There is also a ``validity'' condition to ensure the algorithm is not trivial.)  The consensus problem for a system of $n$ processors is called $n$-consensus, and a consensus algorithm is termed \emph{wait-free} if up to $n-1$ processors can crash in an $n$ processor system, and even so all correctly working processors will decide on the same value.

This paper is part of a larger program to connect the fields of self-assembly and distributed computing.  By reducing models of self-assembly to models of distributed processors, one can apply known distributed computing impossibility results to obtain limits to the power of self-assembly~\cite{sterling}~\cite{mult-nuc-journal}.  The results in Section~\ref{section:simulatethree} of this paper make use of a three-dimensional tile assembly model.  Unlike two-dimensional tiles, 3D tiles have not been implemented in the lab, and it is not clear whether they can be, given current methods.  The objective of the current paper is to explore which distributed objects can (and cannot) be simulated by self-assembling systems, in order to clarify how positive results of distributed computing can apply to self-assembly.

We have organized the rest of the paper as follows.  Section~\ref{section:background} provides further background about tile self-assembly and distributed computing.  Section~\ref{section:simulatemessagepassing} presents a tile assembly simulation of a two-processor message-passing distributed system.  In Section~\ref{section:simulatetwo}, we construct a tile assembly simulation of wait-free 2-consensus.  In Section~\ref{section:simulatethree}, we define a strengthening of the consensus problem, and show that two-dimensional tile assembly systems cannot simulate solutions to it for systems of three or more processes, but three-dimensional tile assembly systems can.  Section~\ref{section:conclusion} concludes the paper and provides suggestions for future research.
\section{Background} \label{section:background}
\subsection{Tile self-assembly background}
Winfree's objective in defining the Tile Assembly Model was to provide a useful mathematical abstraction of DNA tiles combining in solution in a random, nondeterministic, asynchronous manner~\cite{winfree}.  Rothemund~\cite{rothemund}, and Rothemund and Winfree~\cite{programsize}, extended the original definition of the model.  For a comprehensive introduction to tile assembly, we refer the reader to~\cite{rothemund}.  Intuitively, we desire a formalism that models the placement of square tiles on the integer plane, one at a time, such that each new tile placed binds to the tiles already there, according to specific rules.  Tiles have four sides (often referred to as north, south, east and west) and exactly one orientation, \emph{i.e.}, they cannot be rotated.

A tile assembly system $\mathcal{T}$ is a 5-tuple $(T,\sigma,\Sigma, \tau, R)$, where $T$ is a finite set of tile types; $\sigma$ is the \emph{seed tile} or \emph{seed assembly}, the ``starting configuration'' for assemblies of $\mathcal{T}$; $\tau:T \times \{N,S,E,W\} \rightarrow \Sigma \times \{0,1,2\}$ is an assignment of symbols (``glue names'') and a ``glue strength'' (0, 1, or 2) to the north, south, east and west sides of each tile; and a symmetric relation $R \subseteq \Sigma \times \Sigma$ that specifies which glues can bind with nonzero strength.  In this model, there are no negative glue strengths, \emph{i.e.}, two tiles cannot repel each other.

In this paper, we allow for the possibility of errors in binding between tiles.  While, in general, binding errors can cause unplanned configurations to be built, we will make a simplifying assumption that the only binding errors that might occur are \emph{tile blockages}, tile mismatches that prevent any further tiles from binding to the subassembly at which the blockage occurred.  In particular, no erroneously bound tile can be enclosed by tiles that attach later in the process of self-assembly.

Figure~\ref{figure:tilebindingexample} shows a toy example of a tile assembly system, and a possible blockage.

A \emph{configuration of $\mathcal{T}$} is a set of tiles, all of which are tile types from $\mathcal{T}$, that have been placed in the plane, and the configuration is \emph{stable} if the binding strength (from $\tau$ and $R$ in $\mathcal{T}$) at every possible cut is at least 2.  An \emph{assembly sequence} is a sequence of single-tile additions to the frontier of the assembly constructed at the previous stage.  Assembly sequences can be finite or infinite in length.  The \emph{result} of assembly sequence $\overrightarrow{\alpha}$ is the union of the tile configurations obtained at every finite stage of $\overrightarrow{\alpha}$.  The \emph{assemblies produced by $\mathcal{T}$} is the set of all stable assemblies that can be built by starting from the seed assembly of $\mathcal{T}$ and legally adding tiles.  If $\alpha$ and $\beta$ are configurations of $\mathcal{T}$, we write $\alpha \longrightarrow \beta$ if there is an assembly sequence that starts at $\alpha$ and produces $\beta$.  An assembly of $\mathcal{T}$ is \emph{terminal} if no tiles can be stably added to it.

We are, of course, interested in being able to \emph{prove} that a certain tile assembly system always achieves a certain output.  In~\cite{solwin}, Soloveichik and Winfree presented a strong technique for this: local determinism.  An assembly sequence $\overrightarrow{\alpha}$ is \emph{locally deterministic} if (1) each tile added in $\overrightarrow{\alpha}$ binds with the minimum strength required for binding; (2) if there is a tile of type $t_0$ at location $l$ in the result of $\alpha$, and $t_0$ and the immediate ``OUT-neighbors'' of $t_0$ are deleted from the result of $\overrightarrow{\alpha}$, then no other tile type in $\mathcal{T}$ can legally bind at $l$; the result of $\overrightarrow{\alpha}$ is terminal.  A tile assembly system $\mathcal{T}$ is locally deterministic iff every legal assembly sequence of $\mathcal{T}$ is locally deterministic.  Local determinism is important because of the following result.
\begin{theorem}[Soloveichik and Winfree~\cite{solwin}]
If $\mathcal{T}$ is locally deterministic, then $\mathcal{T}$ has a unique terminal assembly.
\end{theorem}
\begin{figure}
\centering
\includegraphics[height=3in]{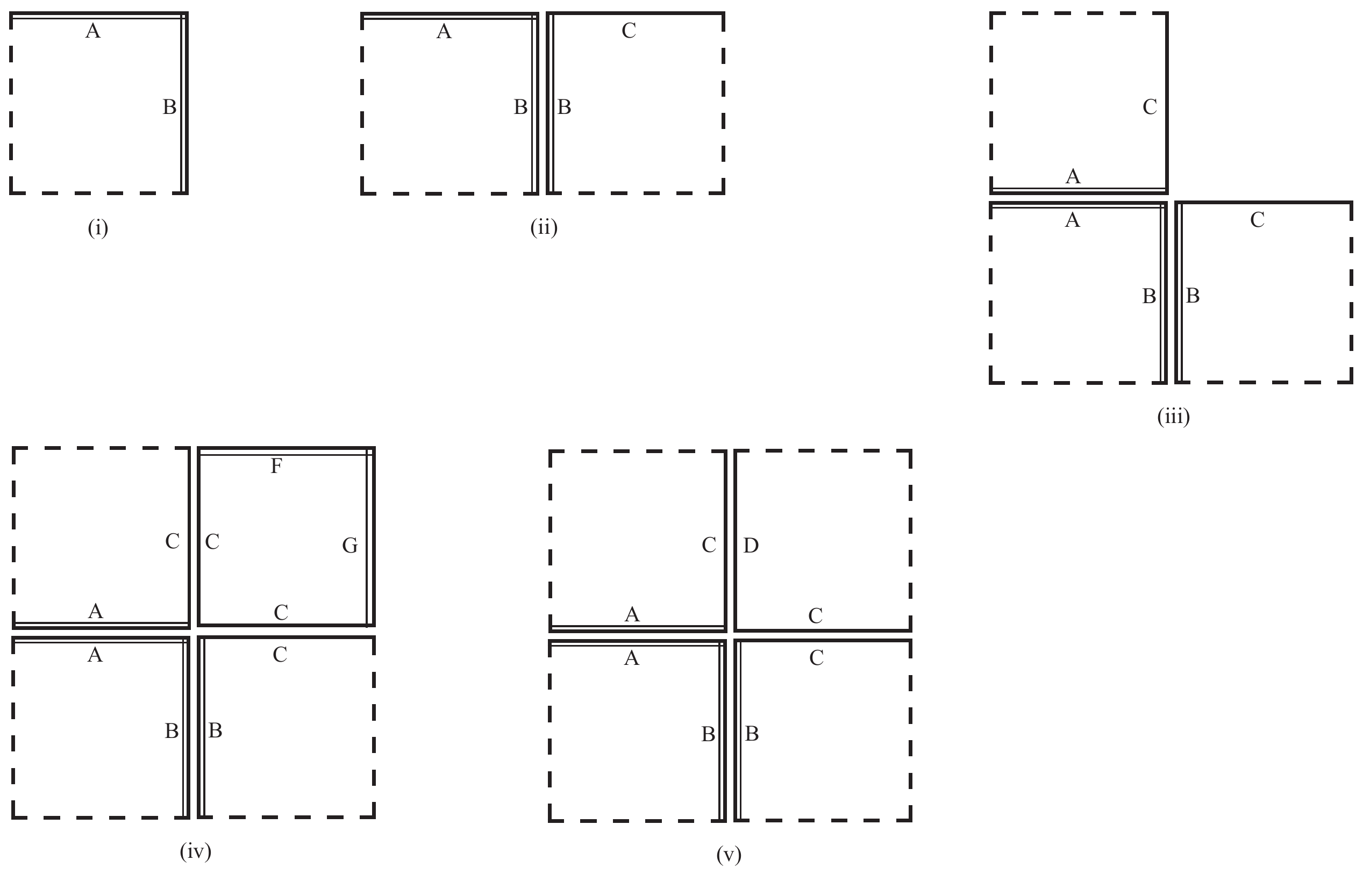}
\caption{A toy example of tile self-assembly.  In (i) a seed tile is placed on a surface.  The dashed sides indicate glue strength zero, and the double-lined sides indicate glue strength two.  In (ii) a tile binds to the east of the seed tile, since the glue strengths and names of the two tiles are complementary.  in (iii) a tile binds to the north of the seed tile for the same reason.  Note that steps (ii) and (iii) could have occurred in reverse order; this is an example of the nondeterminism of self-assembly.  In (iv) a tile binds because its west side (which is strength one, as represented by the single line) and south side (also strength one) bind cooperatively to the tiles already in the configuration, creating a strength two bond.  Figure (v) shows an alternative to (iv), in which a tile binds incorrectly, causing a blockage that prevents the assembly from growing further. }
\label{figure:tilebindingexample}
\end{figure}
\subsection{Distributed computing background}
Distributed computing began as the study of networks of processors, in which each processor had limited local knowledge. However, much of the distributed computing literature now speaks in terms of systems of \emph{processes}, not processors, to emphasize that the algorithms or bounds obtained from the theorem apply to any appropriate collection of discrete processes, whether they are part of the same multicore chip, or spread across a sensor network.

One of the most-studied distributed computing models is the \emph{asynchronous shared memory model}, in which processes are modeled as finite state machines that can read and write to one or more shared memory locations called \emph{registers}.  The model is asynchronous because there is no restriction on when a process will execute its next computation step, except that any nonfaulty process can delay only a finite length of time before taking its next step. We refer the reader to~\cite{attiya} for exposition and technical details of such a model.  We now provide a formal definition of the consensus problem for a distributed system.
\begin{definition}
Let $\mathcal{M}$ be an asynchronous shared memory model such that each process $p_i$ in $\mathcal{M}$ has two special state components: $x_i$, the \emph{input}; and $y_i$, the \emph{decision value}.  Let $V$, the \emph{set of possible decision values}, be a finite set of positive integers.  We require that $x_i \in V$ for all $i$.  For each $i$, $y_i$ starts out containing a null entry, $y_i$ is write-once, and the value written cannot be erased.  A \emph{solution to the consensus problem} is an algorithm that guarantees the following.
\begin{description}
\item[Termination] Every nonfaulty process $p_i$ eventually writes a value to $y_i$.
\item[Agreement] For any $i,j$, if $p_i$ and $p_j$ are nonfaulty and write to $y_i$ and $y_j$, then $y_i=y_j$.  All nonfaulty processes decide on the same decision value.
\item[Validity] For any $i$, if $p_i$ is nonfaulty, then the $y_i$ written by $p_i$ must be the input value $x_j$ for some processor $p_j$.
\end{description}
\end{definition}
We will consider \emph{fault-tolerant consensus}, in which one or more processes can fail in some way.  The simplest type of failure---and the only type we will consider in this paper---is \emph{crash failure}, meaning that at some point a process ceases operation and never takes another step.

It is a classic result of distributed computing that there is no deterministic algorithm that solves the consensus problem in an asynchronous shared memory model, in the presence of even a single crash failure, when the only registers available to the system are read/write registers~\cite{flp}.  One way to overcome that impossibility result is by making the registers more powerful.  These more powerful registers are often called \emph{objects} or \emph{shared objects}, to emphasize they might be a special process segment unto themselves, not just a single memory location.  For a comprehensive introduction to the theory of shared objects with different consensus strengths, see~\cite{hershav}.  We will use only the following key definitions in this paper.
\begin{definition}
Object $O$ \emph{solves wait-free $n$-process consensus} if there exists an asynchronous consensus algorithm for $n$ processes, up to $n-1$ of which might fail (by crashing), using only shared objects of type $O$ and read/write registers.  In distributed system $\mathcal{M}$, $O$ is a \emph{consensus object} if each process in $\mathcal{M}$ can invoke $O$ with a command \texttt{invoke($O$,$v$)}, where $v \in V$ is a possible decision value, and $O$ will ack with command \texttt{return($O$, $v_{out}$)}, where $v_{out} \in V$ is a possible decision value, and $O$ returns the same value $v_{out}$ to all processes that invoke it.  $O$ is an \emph{$n$-consensus object} if $O$ is a consensus object and $n$ is maximal such that $O$ solves wait-free $n$-process consensus.
\end{definition}
Finally, we define the notions of configurations and execution segments of a distributed system.  Intuitively, we consider the events of a system to be the read and write invocations (and acks) performed upon (and returned by) registers by the processors of the system; and configurations and execution segments are built up from the instantaneous state of the system, and events that are then applied to it.
\begin{definition}
Let $\mathcal{M}$ be a distributed system with $n$ processes and one $n$-consensus object $O$.  A \emph{configuration} of $\mathcal{M}$ is an $(n+1)$-tuple $\langle q_1,\ldots,q_n,o \rangle$, where $q_i$ is a state of $p_i$ and $o$ is a state of $O$.  The \emph{events} in $\mathcal{M}$ are the computation steps by the processes, the transmission of a consensus decision value from $O$, and the (possible) crashes of up to $n-1$ of the processes.  A \emph{legal execution segment} of $\mathcal{M}$ is a sequence of form $C_0,\phi_0,C_1,\phi_1,C_2,\phi_2 \ldots$ where each $C_i$ is a configuration, $\phi_i$ is an event, and the application of $\phi_i$ to $C_{i-1}$ results in $C_i$.
\end{definition}
\section{Simulating a two-processor message-passing system in tile self-assembly} \label{section:simulatemessagepassing}
In practice, asynchronous shared memory systems are usually built on top of distributed systems in which processors interact physically by passing messages.  The shared memory model is a convenience for the programmer, as shared memory is easier to reason about than message-passing.  Our method of simulating shared consensus objects is conceptually similar: we will construct tile assembly machinery to simulate a two-processor message-passing distributed system, and then in the next section modify that machinery to simulate a shared 2-consensus object.

Winfree showed how to simulate the behavior of a Turing machine by means of a ``wedge construction'' in tile assembly~\cite{winfree}---essentially, the growing tile assembly grew level by level, with each level simulating one step of a Turing machine execution, the seed tile simulating the initial input, and (if needed) the level expanding one tile longer than the previous level to simulate the Turing machine head moving to a never-before-reached cell of the work tape.  It is this expansion of a level (always occurring in the same direction) that gives the terminal assembly of the simulation a wedge-like appearance.  Lathrop \emph{et al.}~\cite{ccsa} extended this wedge construction so that the notion of ``simulate'' would mean that each row of the wedge records the state of the Turing machine tape after one move of the head, and, if the Turing machine halts, a row of tiles is built along the side of the wedge, so a special ``halting tile'' binds to the base of the wedge, as a marker that the simulation has halted.

We now modify the construction of~\cite{ccsa} to simulate two processors capable of sending each other messages.  The intuition behind the simulation is that we add an input buffer and an output buffer to Winfree's original construction by having every other row of the wedge check to see whether a message has been received, and to incorporate the message if one has arrived; and we simulate the sending and receipt of a constant-size number of messages by using a unique tile that encodes the message to build a ray along the edge of the wedge, toward the intended destination of the message, and a ray back along the edge of the wedge to simulate the \emph{ack}---the character that acknowledges receipt of the message.  A high-level schematic for the construction appears in Figure~\ref{figure:wedge0}.
\begin{figure}
\centering
\includegraphics[height=3in]{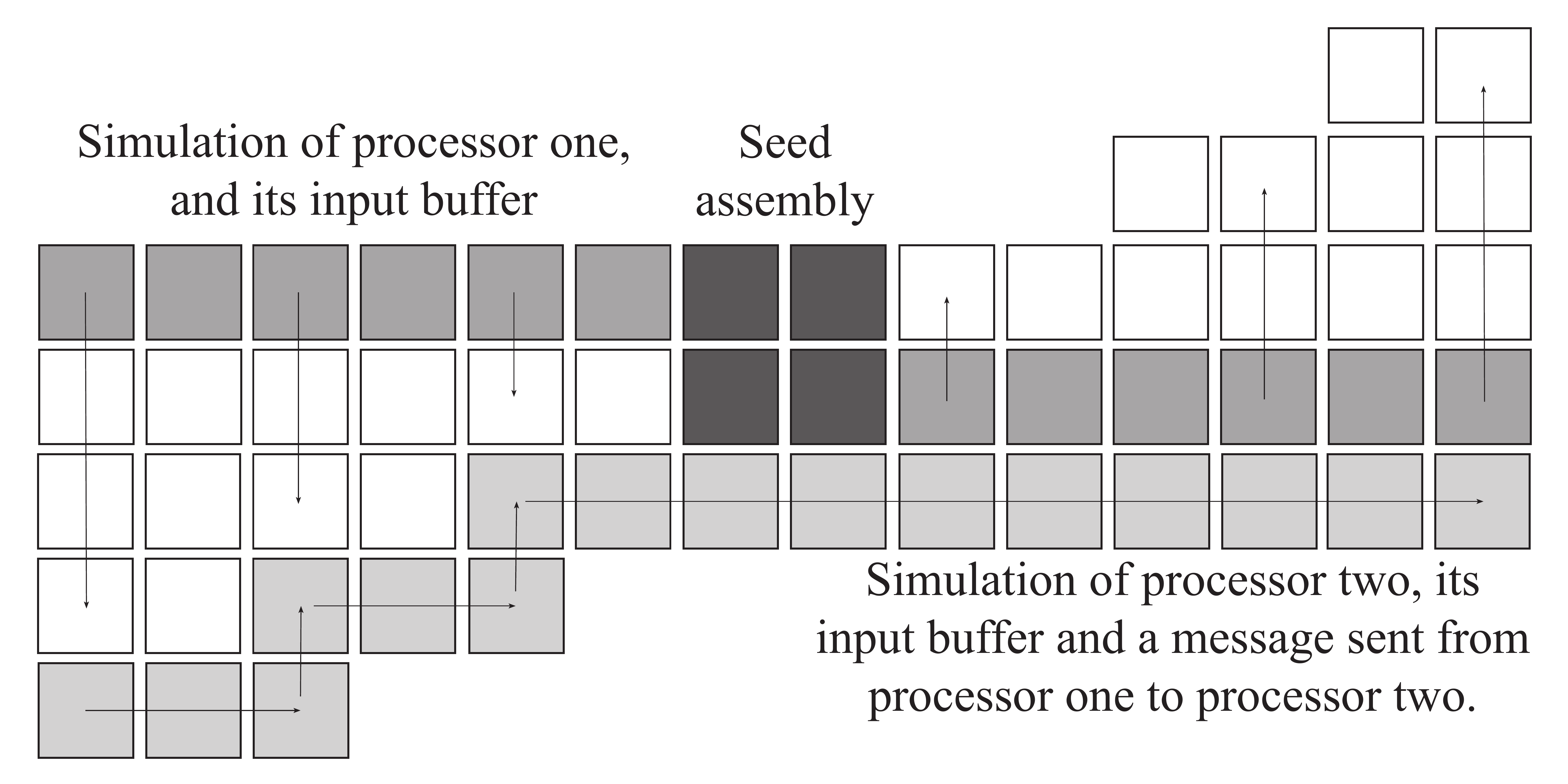}
\caption{Schematic of how to simulate a two-processor message-passing system with two-dimensional tiling.  The green tiles are the initial seed assembly, out of which the white wedges grow.  The wedges simulate the processors.  An additional tile is attached to each level of the wedge, to simulate the processor's inbuffer, as described further in Figure~\ref{figure:wedge}.  To simulate sending a message from processor 1 to processor 2, the wedge on the left sends a message down the side that does not contain the inbuffer.  The transmission of that message is done by the tiles colored light grey.  (The black arrows indicate the order in which the tiles bind during assembly.)  If processor 2 sent a message to processor 1, we could simulate this in tiles by constructing a similar ``grey'' transmission from the northeast edge of the wedge simulating processor 2, westward toward processor 1 (essentially a mirror image of the grey tiles in the schematic).}
\label{figure:wedge0}
\end{figure}
\begin{theorem} \label{theorem:tilesimulation}
For any two-processor message-passing model $\mathcal{M}$ of distributed computing in which each process runs forever and sends messages of constant size, there is a tile assembly system $\mathcal{T}$ (in the standard Winfree-Rothemund Tile Assembly Model) that simulates $\mathcal{M}$ in two dimensions.
\end{theorem}
\begin{proof}
Fix $\mathcal{M}$ a two-processor asynchronous message-passing model of distributed computing, such that processors $p_1$ and $p_2$ send each other messages of size bounded by some constant $k$.  Winfree has shown that the Tile Assembly Model is Turing Universal, and in particular can simulate a Turing Machine using a construction that grows in the shape of a ``wedge''~\cite{winfree}.  If neither $p_1$ nor $p_2$ ever sends (or receives) messages, then let $\mathcal{T}_1$ and $\mathcal{T}_2$ be tile assembly systems that simulate $p_1$ and $p_2$.  Then we can construct a tile assembly system $\mathcal{T}^*$ with a seed out of which $\mathcal{T}_1$ grows in one direction, and $\mathcal{T}_2$ grows in another.

So assume that at least one of $p_1,p_2$ sends messages to the other.  To simulate processors checking their inbuffers, we modify the wedge construction, so that (1) messages sent to the processor simulated by the wedge can ``crawl up'' the side of the wedge until they reach the frontier; and (2) every other row of tiles in the wedge allows sent messages to bind with that row with probability bounded away from zero.  Therefore, the construction remains two-dimensional, and, since both $p_1$ and $p_2$ run forever, the message will be delivered with probability one.  Figures~\ref{figure:wedge} and~\ref{figure:wedge2} illustrate how the modified wedge construction works.

Since there are only $k$-many messages, we encode each message with a unique tile.  We simulate the sending of messages from $p_1$ to $p_2$ by allowing tiles to bind along one ``edge'' of the tile assembly, and simulate the sending of messages from $p_2$ to $p_1$ along the other ``edge.''

It is important that the Tile Assembly Model assumes that at each stage of assembly, the tile that binds---and the location it binds to---are chosen uniformly at random.  It takes several tiles to build one row of a wedge, but just one tile to propagate a message up one row, closer to the frontier.  So, almost surely, the message will have unboundedly many opportunities to incorporate its information into the ``inbuffer row''; hence, almost surely, the message will be delivered.
\end{proof}
\begin{figure}
\centering
\includegraphics[height=6in]{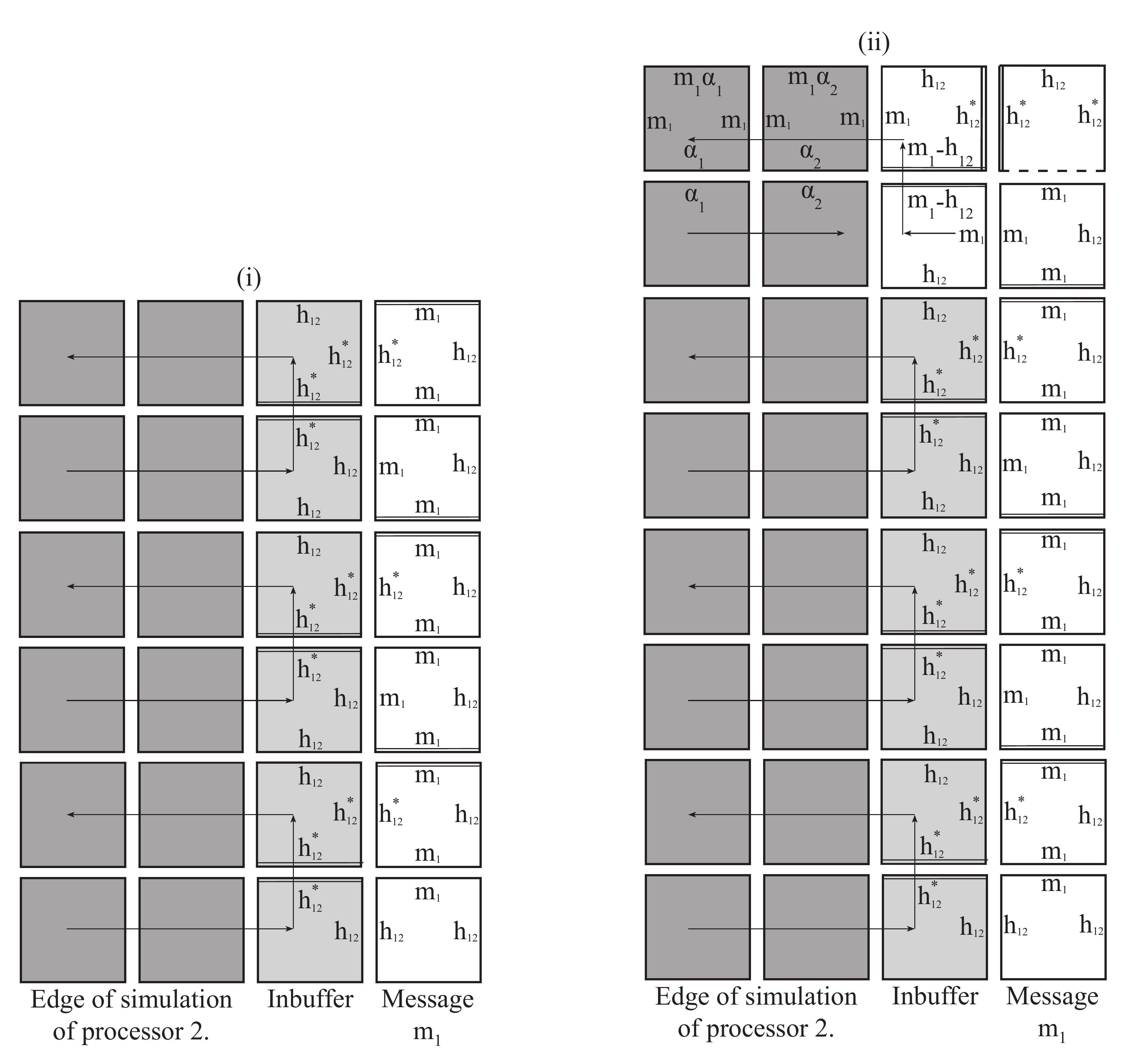}
\caption{This figure illustrates the simulation of a processor, its inbuffer, and the receipt of a message.  The dark grey tiles are the edge of a ``wedge'' that simulates processor 2 (in the simulation of a two-processor system).  The light grey tiles represent the inbuffer of processor 2.  In this diagram, every other row of the wedge is a row that is built from the inbuffer tile, to the west.  The arrows show the direction that tiles bind in.  More generally, the wedge construction does not need to ``check'' the inbuffer at every other row, as long as it checks infinitely often.  Figure (i) shows how message $m_1$ attaches along the east side of the inbuffer, and has opportunities due to double bonds to bind to the north of the most recent inbuffer tile.  In Figure (ii), $m_1$ is successfully transferred into the inbuffer, which then means that the next row that checks the contents of the inbuffer transmits $m_1$ to (potentially) all tiles in the wedge, as shown.  The farthest northeast tile in Figure (ii) changes the column to an inbuffer column, in case later messages need to be transferred from the east to the wedge growing to the west.}
\label{figure:wedge}
\end{figure}
\begin{figure}
\centering
\includegraphics[height=7in]{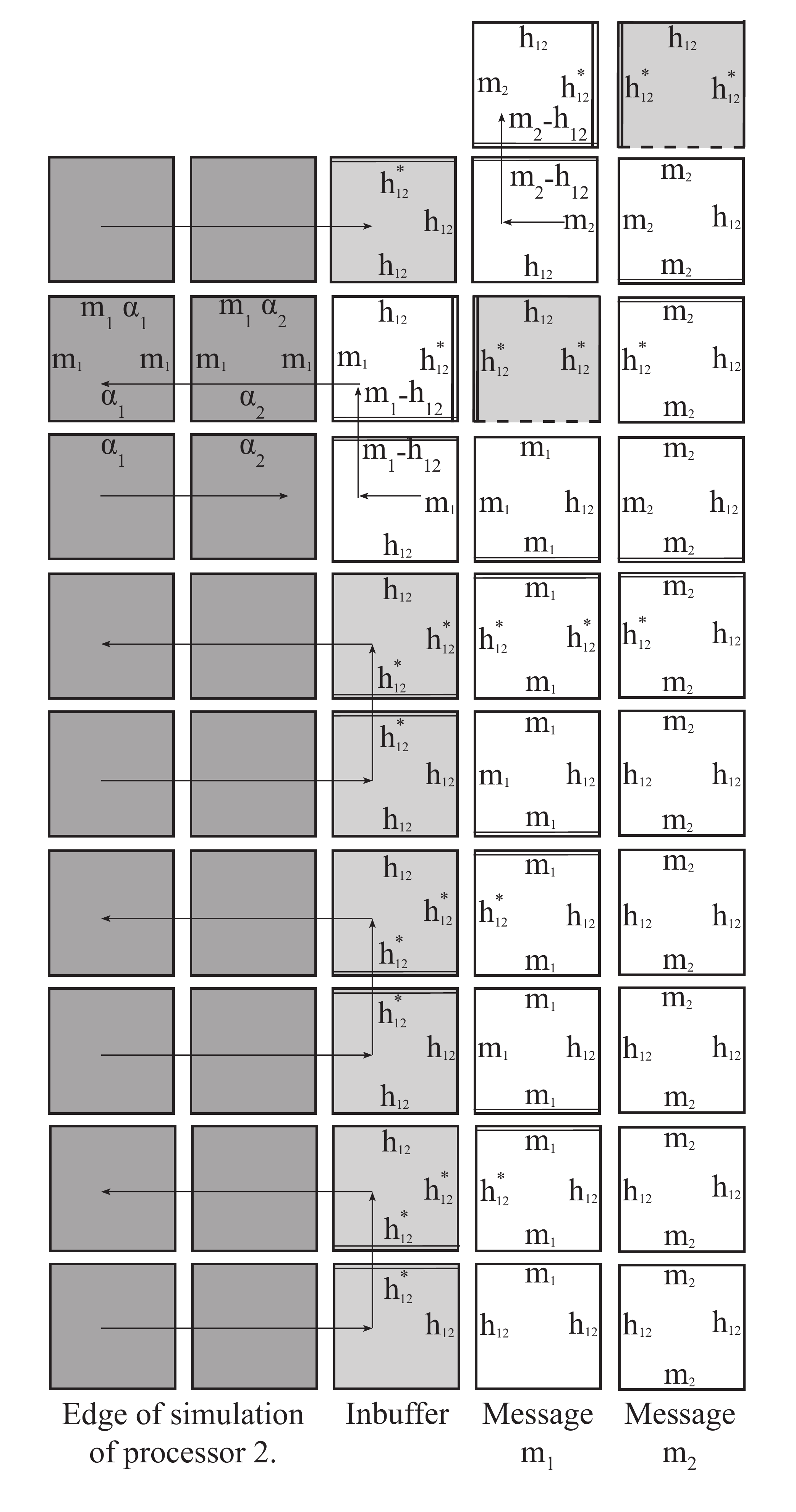}
\caption{This figure continues the construction from Figure~\ref{figure:wedge}(ii), by showing how a second message, named $m_2$, gets transferred to the ``first message'' column of the wedge simulating processor 2.  The mechanism shown in Figure~\ref{figure:wedge} will then transfer $m_2$ to the inbuffer.  Note that the furthest northeast tile sets up the configuration so any further messages sent can be transferred west, as $m_2$ was.}
\label{figure:wedge2}
\end{figure}
\section{Simulating 2-consensus in two dimensions} \label{section:simulatetwo}
We turn now to the simulation in tile assembly of shared objects that solve wait-free consensus problems.  We will discuss in this section how to simulate 2-consensus by a two-dimensional tile assembly system, and in the next section show how to extend the simulation to more processes, and to a third tiling dimension.

Intuitively, the main objective of this section is to construct a tile assembly system that behaves as shown in Figure~\ref{figure:simulation}: there are three ``modules'' or subassemblies: one ($\pi_1$) to simulate $p_1$, another ($\pi_2$) to simulate $p_2$, and a third to simulate a 2-consensus object.  In addition, there are ``rays'' of tiles built from $\pi_1$ and $\pi_2$, to simulate the values that $p_1$ and $p_2$ are writing to the 2-consensus object; and ``rays'' of tiles built in response, to simulate the acks received by $p_1$ and $p_2$ after the writes conclude.

First, we define what it means for a tile assembly system to simulate a distributed system.  Informally, we want to capture the following: each configuration of the distributed system is simulated by a tile configuration; all legal behavior of the distributed system is modeled by legal behavior of the tile assembly system; conversely, if it is impossible to take a step in the distributed system, the tile assembly system cannot simulate taking that step; and the tile assembly system halts, or a section of it becomes blocked, if and only if the distributed system reaches a halting configuration, or the corresponding processor crashes.
\begin{definition}
Tile assembly system $\mathcal{T}$ \emph{simulates} distributed system $\mathcal{M}$ if:
\begin{enumerate}
\item There is a 1-1 mapping $h$ from configurations of $\mathcal{M}$ to stable tile configurations of $\mathcal{T}$.
\item If $C_0,\phi_0,C_1,\phi_1 \cdots C_i,\phi_i$ is a legal execution segment of $\mathcal{M}$, then $h(C_0) \longrightarrow h(C_1) \longrightarrow \cdots \longrightarrow h(C_i)$ is a legal tile assembly sequence in $\mathcal{T}$.
\item If there is no legal execution segment from $C_0$ to $C_1$ in $\mathcal{M}$, then there is no legal tile assembly sequence in $\mathcal{T}$ such that $h(C_0) \longrightarrow h(C_1)$.
\item Let $C$ be a configuration of $\mathcal{M}$ and $\mathcal{C}$ be a set of configurations of $\mathcal{M}$.  If $\mathcal{M}$ is such that, upon achieving configuration $C$, it must eventually achieve some configuration $C' \in \mathcal{C}$ unless a process crashes, then $\mathcal{T}$ is such that, if it ever reaches $h(C)$ then it must achieve $h(C')$ for some $C' \in \mathcal{C}$ unless there is a tile blockage.  (Note that $\mathcal{C}$ may be an infinite set.)
\item If $C_0,\phi_0,C_1,\phi_1$ is a legal execution segment in $\mathcal{M}$, and the event $\phi_0$ is the crash failure of a previously correct process, then $h(C_1)$ contains one more tile blockage binding error than $h(C_0)$ contains.
\item If in configuration $C$ of $\mathcal{M}$ all processes have halted, then in tile configuration $h(C)$ of $\mathcal{T}$ there are no further locations to which tiles can bind stably, \emph{i.e.}, $h(C)$ is terminal.
\end{enumerate}
\end{definition}
\begin{figure}
\centering
\includegraphics[height=3.5in]{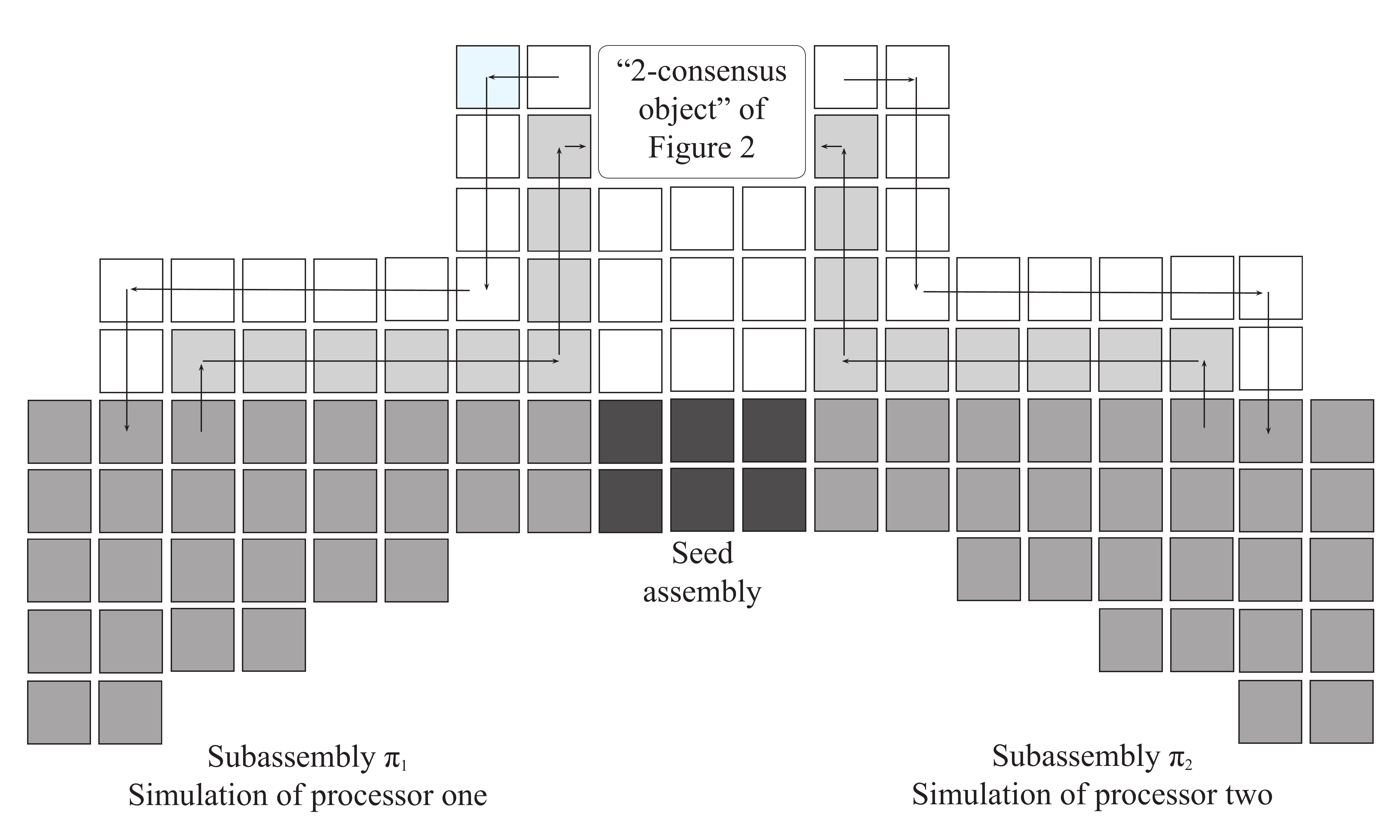}
\caption{Two (yellow) subconfigurations $\pi_1$ and $\pi_2$ growing from a common (green) seed assembly.  (The arrows indicate the order in which tiles bind to the assembly.)  They communicate with the configuration shown in Figure~\ref{figure:testandset}, by means of (grey) ``message'' tiles they send, and (white) ``ack'' tiles they receive back.  Intuitively, $\pi_1$ and $\pi_2$ simulate two processes in a distributed system, communicating with a 2-consensus object.}
\label{figure:simulation}
\end{figure}
The main objective of this section is to prove the following theorem.
\begin{theorem} ~\label{theorem:2Dsimulation}
 Let $\mathcal{M}$ be an asynchronous message passing model of distributed computing with two processes and one 2-consensus object $O$, such that $p_1$ and $p_2$ send no messages to each other, and such that each process invokes $O$ at most once.  Then there is a tile assembly system $\mathcal{T}$ that simulates $\mathcal{M}$.
\end{theorem}
To prove Theorem~\ref{theorem:2Dsimulation} we first exhibit a tile configuration that can simulate a 2-consensus object.
\begin{lemma} \label{lemma:testnset}
Fix $V$ any finite set.  There is a tile configuration $\rho$ that contains a binding location $l$ with the following properties: (1) the only tiles that bind at $l$ have nonzero glue strenghs at either the south, north and west sides, or the south, north and east sides; (2) any tile that binds at $l$ will have a glue name taken from $V$ on either its east or west side; (3) if glue name $v \in V$ is on the tile that binds at $l$, then the name of the tile's north glue will be ``$Ackv$,'' and $\rho$ will build a ray transmitting the glue name $Ackv$ to its west and east.
\end{lemma}
\begin{figure}
\centering
\includegraphics[height=5.5in]{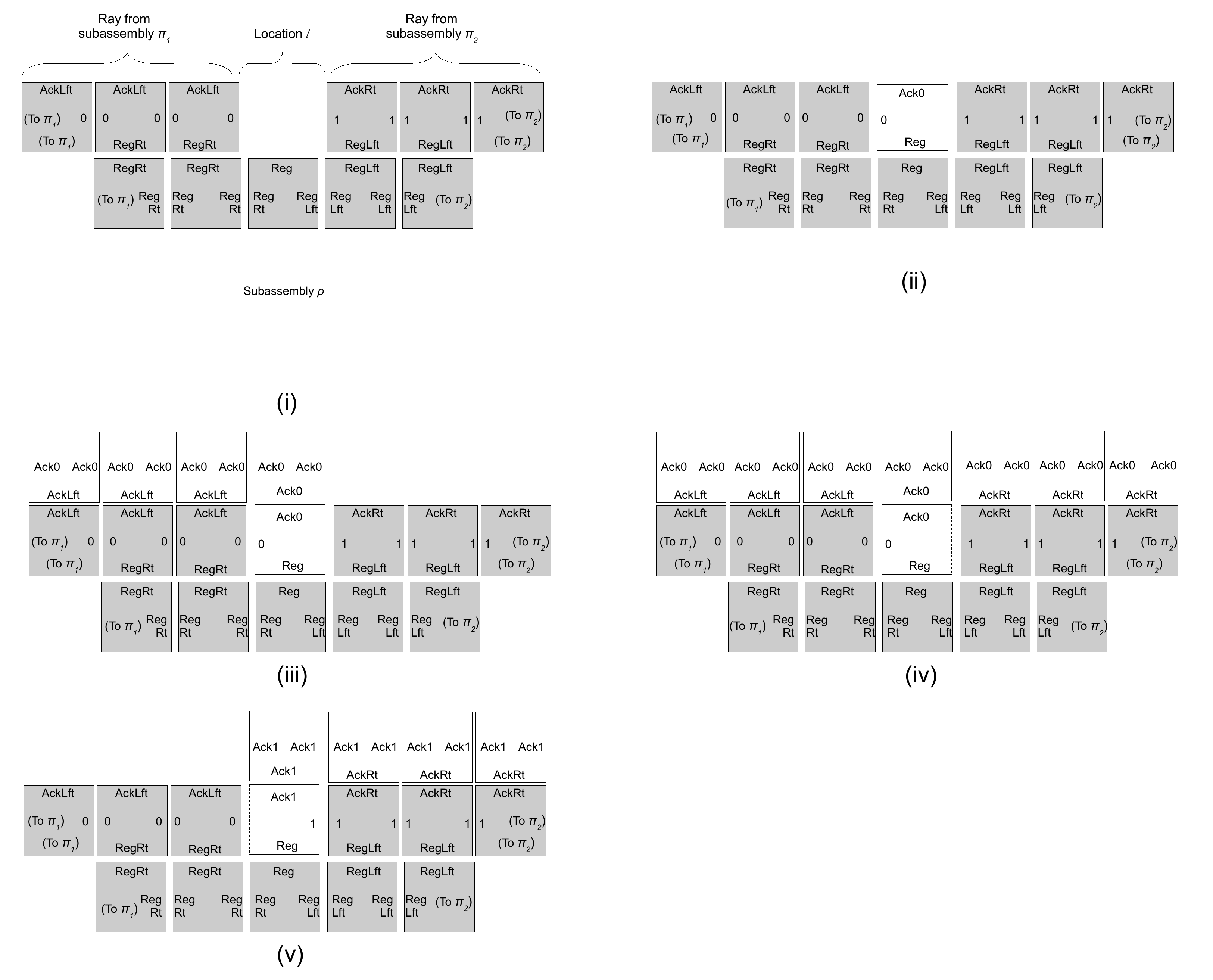}
\caption{Diagram of the simulator of a 2-consensus object used to prove Theorem~\ref{theorem:2Dsimulation}.  At stage (i), rays from subassemblies $\pi_1$ and $\pi_2$ approach location $l$.  At stage (ii), a tile binds at $l$, in this case deciding in favor of the input value of subassembly $\pi_1$.  At stage (iii), the simulator sends an ack to $\pi_1$.  At stage (iv), the simulator sends an ack to $\pi_2$.  ``Stage'' (v) demonstrates the alternative: a decision tile has bound in favor of the initial value of $\pi_2$, and the simulator has acked to $\pi_2$.}
\label{figure:testandset}
\end{figure}
Let $V$ be the set of decision values of the consensus problem we want to simulate.  Then Figure~\ref{figure:test-and-set-tileset} exhibits a tileset of $6 \cdot |V|+5$ tile types that assembles $\rho$ with the desired properties.
\begin{proof}[Lemma~\ref{lemma:testnset}]
Fix a finite set $V$ and a locally deterministic tile assembly system $\mathcal{T}$ that builds all but the top row of the tile configuration shown in Figure~\ref{figure:testandset}(i).  For each $v \in V$, add a tile as shown in the top row of Figure~\ref{figure:testandset}(i), such that $v$ appears in each location a ``0'' or ``1'' appears in the figure.  If a tile with value $v \in V$ as its west glue name binds as a decision tile, as in Figure~\ref{figure:testandset}(ii), then $\rho$ builds an ack as shown in Figure~\ref{figure:testandset}(iii).  Similarly, a tile with value $v' \in V$ as its east glue, $\rho$ builds an ack as shown in Figure~\ref{figure:testandset}(v).  Without loss of generality, assume the decision tile binds with the value transmitted from the west of $\rho$.  When tiles eventually arrive from the east, $\rho$ will build an ack on their north sides, as shown in~\ref{figure:testandset}(iv).  Figure~\ref{figure:test-and-set-tileset} exhibits an explicit tileset of $6 \cdot |V| +5$ tiles for such a tile configuration $\rho$.

Once the decision tile binds, the tileset for $\rho$ is locally deterministic.  Hence, unless a blockage error occurs, $\rho$ builds appropriate acks to both east and west, regardless of the order in which the tiles bind.
\end{proof}
\begin{figure}
\centering
\includegraphics[height=5.5in]{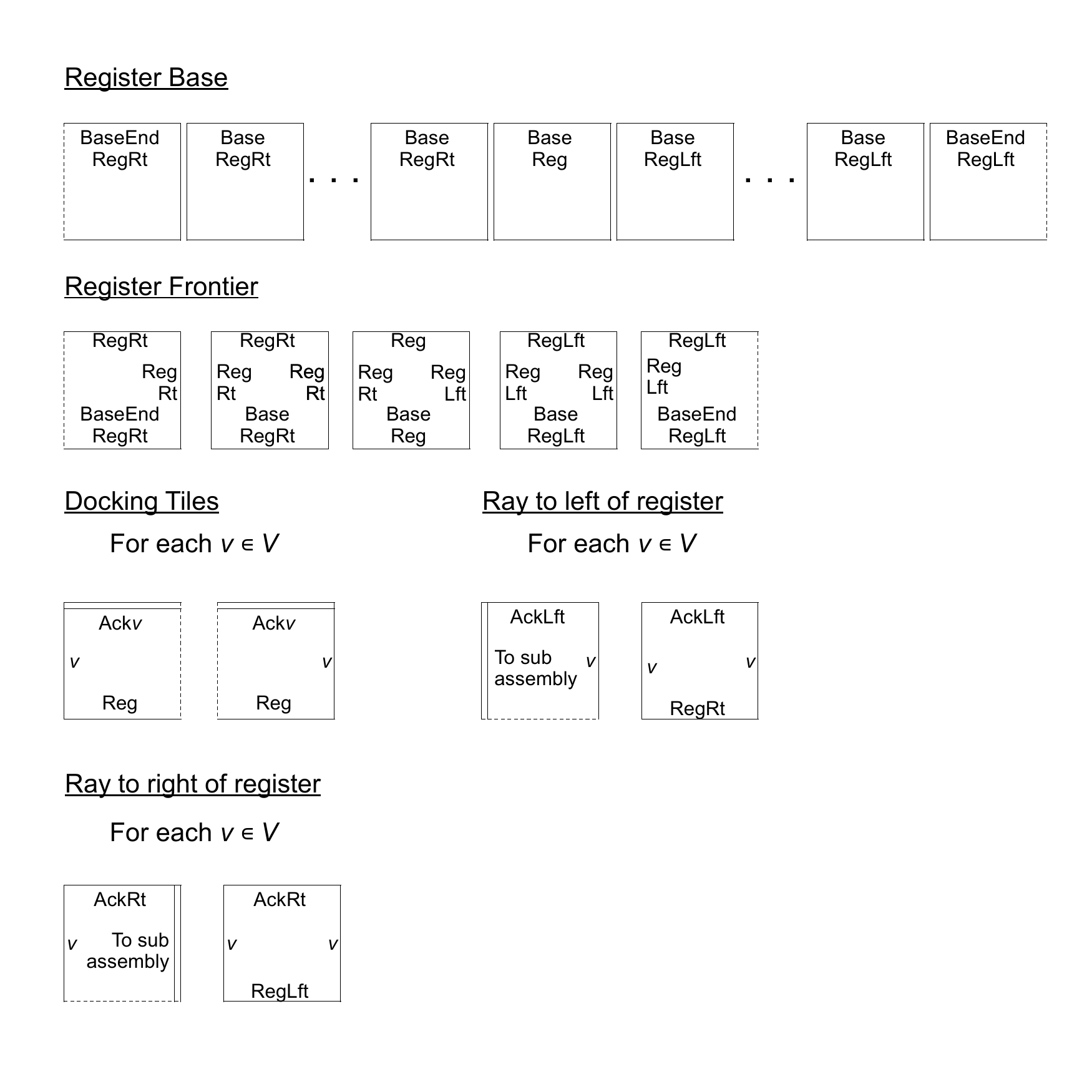}
\caption{A tileset of $6 \cdot |V|+5$ tile types that assembles $\rho$ that will behave according to Figure~\ref{figure:testandset}.  It is important for the proof that this tileset is locally deterministic, except for the choice of tile that binds at location $l$, the decision point.} \label{figure:test-and-set-tileset}
\end{figure}
\begin{proof}[Theorem~\ref{theorem:2Dsimulation}]
Fix $\mathcal{M}$, a system of distributed computing per the theorem statement, let $V$ be the set of possible input values to the consensus problem that $O$ might be invoked to solve, and let $k=|V|$.  We will design $\mathcal{T}$ so it builds the structure shown in Figure~\ref{figure:simulation}.

Using the machinery from the proof of Theorem~\ref{theorem:tilesimulation}, let $\mathcal{T}_1$ be a tile assembly system that simulates $p_1$, $\mathcal{T}_2$ a tile assembly system that simulates $p_2$, and $\mathcal{T}_3$ a tile assembly system that constructs $\rho$, a simulator of a 2-consensus object, per Lemma~\ref{lemma:testnset}.  Without loss of generality, we can require $\mathcal{T}_1$ to build in a direction opposite to $\mathcal{T}_2$ (west and east), and for each to build in such a way that they simulate the growth of the Turing machine tape by extending one additional tile to the south at every other column of the construction.

Begin with tile assembly system $\mathcal{T}^*$ as in the proof of Theorem 2: the seed assembly of $\mathcal{T}^*$ is such that $\mathcal{T}_1$ grows out of it in one direction, and $\mathcal{T}_2$ grows out of it in another, as in Figure 2.  We will build $\mathcal{T}$ by extending $\mathcal{T}^*$ to include the simulation of communication between those two processes and the consensus object.

Extend $\mathcal{T}^*$ to $\mathcal{T}$ by adding, for each of the $k$ possible decision values, unique tiles of the form shown in Figure~\ref{figure:testandset}(i), so $\pi_1$ and $\pi_2$ can send rays to $\rho$, and receive acks as shown in Figure~\ref{figure:testandset}(iv).  (The diagram shows a simulation of binary consensus, with only glue names ``0'' and ``1.''  We replace those glue names with a total of $k$ distinct names.)  We also include in $\mathcal{T}$ the tiles needed so that a subconfiguration that simulates a process can send, and receive, one tile's worth of information.
To show $\mathcal{T}$ simulates $\mathcal{M}$, let $h_1$ be a map that testifies that $\mathcal{T}_1$ simulates $p_1$ (on all process steps except invocations to $O$ and responses from $O$), and $h_2$ testify that $\mathcal{T}_2$ simulates $p_2$ (similarly, without communication with $O$).  Let $\alpha$ be a configuration of $\rho$ in which the base of $\rho$ has been completely built, \emph{i.e.}, everything but the top row of Figure~\ref{figure:testandset}(i).  Then for $C=\langle q_1,q_2,o \rangle$ a configuration of $\mathcal{M}$, we have the following cases.
\begin{enumerate}
\item If neither $p_1$ nor $p_2$ has yet invoked $O$, let $h(C)=h_1(q_1) \cup h_2(q_2) \cup \alpha$.
\item If $p_1$ or $p_2$ (or both) have invoked $O$, but $O$ has not responded to either yet, let $h(C)=h_1(q_1) \cup h_2(q_2) \cup \alpha \cup \beta$, where $\beta$ is the blue tile configuration of Figure~\ref{figure:simulation} that simulates the transmission of messages from $p_1$ or $p_2$ or both, as appropriate.  As the messages have not arrived at $O$, no tiles in $\beta$ reach the top of subassembly $\rho$.
\item If $O$ has responded to exactly one process, say $p_1$ WLOG, let $h(C)=h_1(q_1) \cup h_2(q_2) \cup \alpha \cup \beta$, where $\beta$ is a configuration as in Figure~\ref{figure:testandset}(iii), with the appropriate decision value replaced for glue name ``0.''
\item If $O$ has responded to both processes, say in order $p_1$ and then $p_2$ WLOG, let $h(C)=h_1(q_1) \cup h_2(q_2) \cup \alpha \cup \beta$, where $\beta$ is a configuration as in Figure~\ref{figure:testandset}(iv), with the appropriate decision value replaced for glue names ``0'' and ``1.''
\item If $q_1$ is of form \texttt{return($O$,$v$)} for $v \in V$, let $C'$ be the configuration immediately preceding $C$ in the execution segment.  Define $h(C)$ by adding tiles to $h(C')$ so a tile with the glue name corresponding to decision value $v$ is stably placed at the northeast corner of $\pi_1$, for incorporation into the simulation of $p_1$.  Define $h(C)$ similarly if $q_2$ is of form \texttt{return($O$,$v$)}, except that the tile is stably placed at the northwest corner of $\pi_2$.
\end{enumerate}
Since we are assuming each of $p_1$ and $p_2$ invokes $O$ at most once, a simple induction argument confirms that $\mathcal{T}$ simulates $\mathcal{M}$ via $h$.
\end{proof}
\section{Simulating distributed systems with three or more processes} \label{section:simulatethree}
One critical difference between ``classical'' distributed systems and tile assembly systems is that sending messages---and writing to shared memory locations---in a distributed system does not affect the future computation resources available to processes; whereas in a tile assembly system, the tiles placed on the plane to simulate such operations may ``box in'' other subassemblies, so they cannot grow beyond some point, due to tile blockages.  Put another way, systems of distributed processes have multidimensional resources: each process computes using its own set of resources, and message-passing takes place via a different set of resources.  By contrast, every tile operation self-assembles using the same shared resource: the surface.  It is, therefore, not surprising that to simulate this resource-independence of distributed systems, tile assembly systems require multiple surfaces, that is to say, three spatial dimensions.  Formalizing that is the main objective of this section.

We now generalize the ``traditional'' consensus problem in distributed computing, to make explicit an assumption that is perhaps universal in the distributed computing literature: invocation of a consensus object does not increase the likelihood that a processor will fail.
\begin{definition}
Let $\mathcal{M}$ be a system of distributed processes such that all correctly-operating processes will run forever.  A solution to the \emph{Consensus Subroutine Problem} for $\mathcal{M}$ is an algorithm $A$ that solves the wait-free consensus problem for $\mathcal{M}$ in such a way that no process increases its likelihood of crash failure by calling $A$ as a subroutine.
\end{definition}
If $\mathcal{M}$ has only two processes, then a tile assembly system of the form in Figure~\ref{figure:simulation} will solve the Consensus Subroutine Problem for $\mathcal{M}$, because there is enough space on the surface for $\pi_1$, $\pi_2$ and the messages to and from $\rho$ to assemble without interfering.  For $\mathcal{M}$ with three or more processes, the situation is different.  If we attempt to simulate processes that run forever in two dimensions, the tile configurations that simulate the processes must grow without bound.  This will cause destructive collision with tiles attempting to communicate information with an $n$-consensus object.
\begin{theorem} \label{theorem:impossibility}
Let $\mathcal{M}$ be a system of $n$ distributed processes ($n \geq 3$) and one $n$-consensus object, such that all correctly-working processes will run forever.  Then no two-dimensional tile assembly system can simulate a solution of the Consensus Subroutine Problem for $\mathcal{M}$.
\end{theorem}
\begin{proof}
Let $\mathcal{M}$ be a distributed system with three processes, and otherwise as given in the theorem statement.  Any $\mathcal{T}$ that simulates $\mathcal{M}$ must be able to simulate $\mathcal{M}$ even if any $n-1$ processes of $\mathcal{M}$ crash.  So each tile configuration that simulates a process $p_i \in \mathcal{M}$ must grow independently of other such configurations, else a tile blockage simulating the crash of $p_i$ might also block $p_j$ where $i \neq j$.  Let $\pi_1$, $\pi_2$, $\pi_3$ be the independent tile subconfigurations simulating $p_1$, $p_2$, and $p_3$, respectively.

In order to simulate wait-free 3-consensus, there must be a particular location $l$ on the tiling surface where two processes (WLOG say $p_1$ and $p_2$) agree on a common value.  But then consider an execution segment of the consensus algorithm in $\mathcal{M}$ where $O$ sends a common value to $p_1$ and $p_2$, but neither $p_1$ nor $p_2$ receives that information until $p_3$ has decided on its own final decision value.  (Without loss of generality, we can assume that $l$ is located further from $\pi_3$ than from either $\pi_1$ or $\pi_2$.)  To simulate this in tiling, $\pi_3$ must obtain the information about the correct decision value somehow, and there is no guarantee the information is located anywhere except on the surface somewhere between $\pi_1$ and $\pi_2$, as it has not reached either $\pi_1$ or $\pi_2$ yet.  Configuration $\pi_3$ must either decide on a value without consulting either $\pi_1$ or $\pi_2$, or obtain the information from location $l$.  If $\pi_3$ decides independently, it could decide on the wrong value.  However, to obtain the information from $l$, it must build a transmission ray and an ack ray to and from $l$, which means that the tiles it places forever limit the growth of either $\pi_1$ or $\pi_2$, or both.  (See Figure~\ref{figure:blocking} for a visual representation.)  Any tile assembly system that limits the growth of $\pi_1$, $\pi_2$ or $\pi_3$ cannot solve the Consensus Subroutine Problem for $\mathcal{M}$.
\begin{figure}
\centering
\includegraphics[height=2.75in]{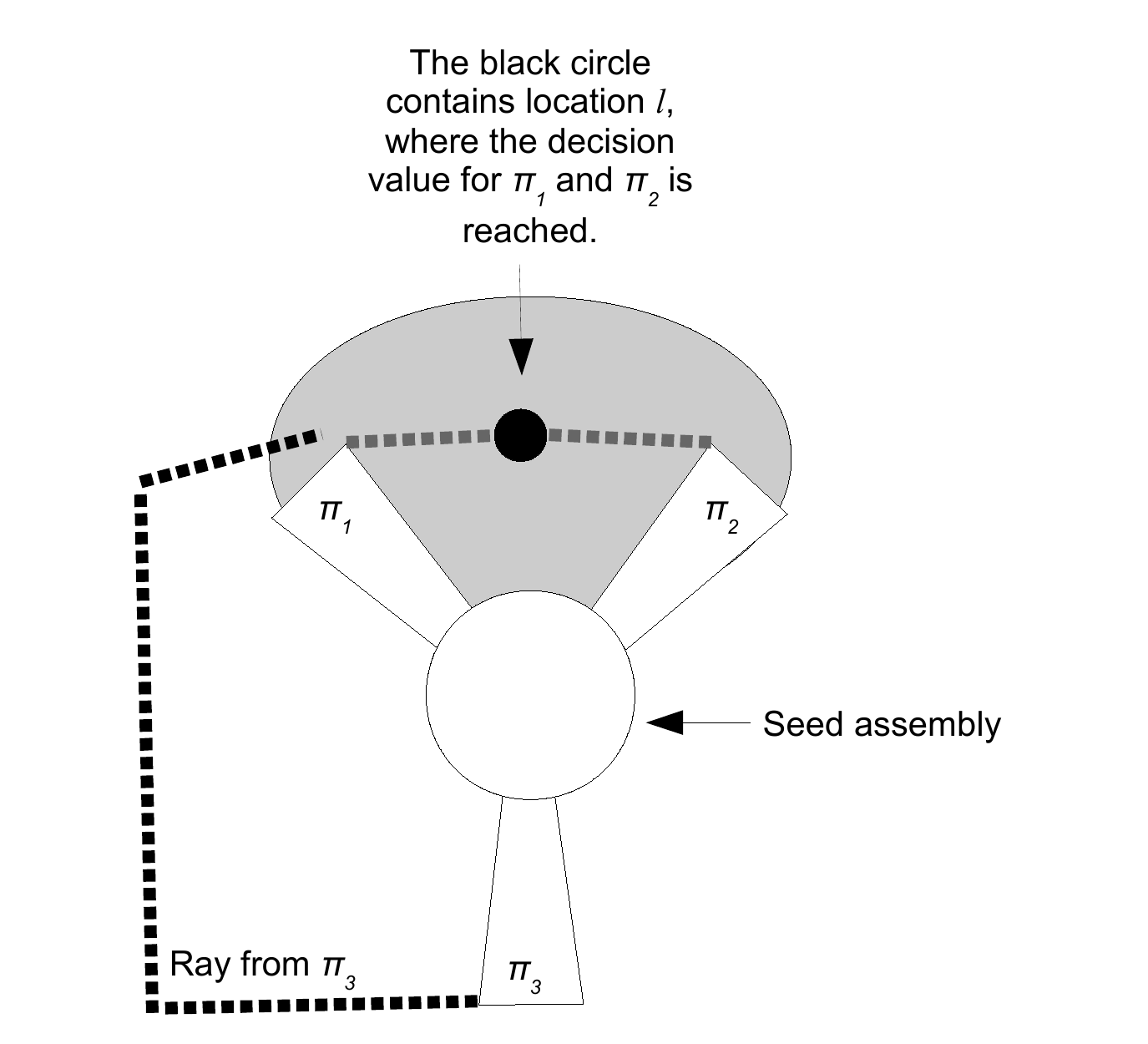}
\caption{For $\pi_1$ and $\pi_2$ to agree on a decision value, they must agree on a value at some location $l$.  If the information at $l$ is tiled outside the blue area before $\pi_1$ and $\pi_2$ incorporate it into their simulation of $p_1$ and $p_2$, it will block further progress of (at least) one of $\pi_1$, $\pi_2$.  So to simulate an execution where $p_1$ and $p_2$ take no further steps until $p_3$ has decided, the decision value must lie inside the blue area.  But if $\pi_3$ sends a ray into the blue area to obtain the decision value, it permanently blocks growth of $\pi_1$ or $\pi_2$ (the figure shows $\pi_3$ blocking $\pi_1$).  However, $\pi_3$ has no other way to make progress if $\pi_1$ and $\pi_2$ agree on a value, and then no more tiles bind to $\pi_1$ or $\pi_2$ until $\pi_3$ has committed to a decision.  (The point is information-theoretic.  If the ray from $\pi_3$ connects to a ray from the red circle before a ray from the red circle connects to $\pi_1$ or $\pi_2$, then $\pi_1$ or $\pi_2$ will be forever unable to send information \emph{through} that tile configuration.)} \label{figure:blocking}
\end{figure}

Given a distributed system with four or more processes, we can just consider the Consensus Subroutine Problem for executions in which all but three processes crash.  We have shown there is no two-dimensional tiling simulation of a solution to the Consensus Subroutine Problem for three processes.  Therefore, there is no such simulation for a system with any number of processes $n \geq 3$.
\end{proof}
To focus on analysis of fault-tolerance of the consensus problem in a distributed system, models of distributed computing tend to require that the consensus object itself be nonfaulty, even though processes may fail.  We will carry over that requirement into our tiling models: we do not permit tile blockages at location $l$ in the construction $\rho$, nor on the ``approach pad'' of the three-dimensional configurations we are about to present.  Therefore, our results show that if all tiles in these configurations bind correctly, they simulate objects capable of solving wait-free $n$-consensus for any finite $n$.

Kao and Ramachandran have proposed a ``natural'' extension of the Winfree-Rothemund Tile Assembly Model from two dimensions to three dimensions~\cite{3D}, in which self-assembling DNA cubes are constructed like folding boxes from two-dimensional DNA tiles.  Figure~\ref{figure:boxfold} illustrates how such a cube could be ``derived'' from planar tiles.  A cubic tile in this model would then be a 4-tuple $T_N \times T_S \times T_W \times T_E$, where each $T_i$ (for $i \in \{N,S,E,W\}$) encodes the glue information needed to ``sew'' the cubic tile together, as in Figure~\ref{figure:boxfold}, and also the glue type on each respective cube face, available to bind to other cubic tiles.  We refer the reader to~\cite{3D} for the formalisms of this model, and for discussion of algorithms and complexity measures in the model.  For our purposes, it is sufficient that there is a (theoretical) way to construct self-assembling cubes whose faces carry glue information, and these cubes bind together following similar rules as in two dimensions: all binding is error-free or a blockage, and bonds are stable if they produce a binding strength of 2 or greater.

\begin{figure}
\centering
\includegraphics[height=3in]{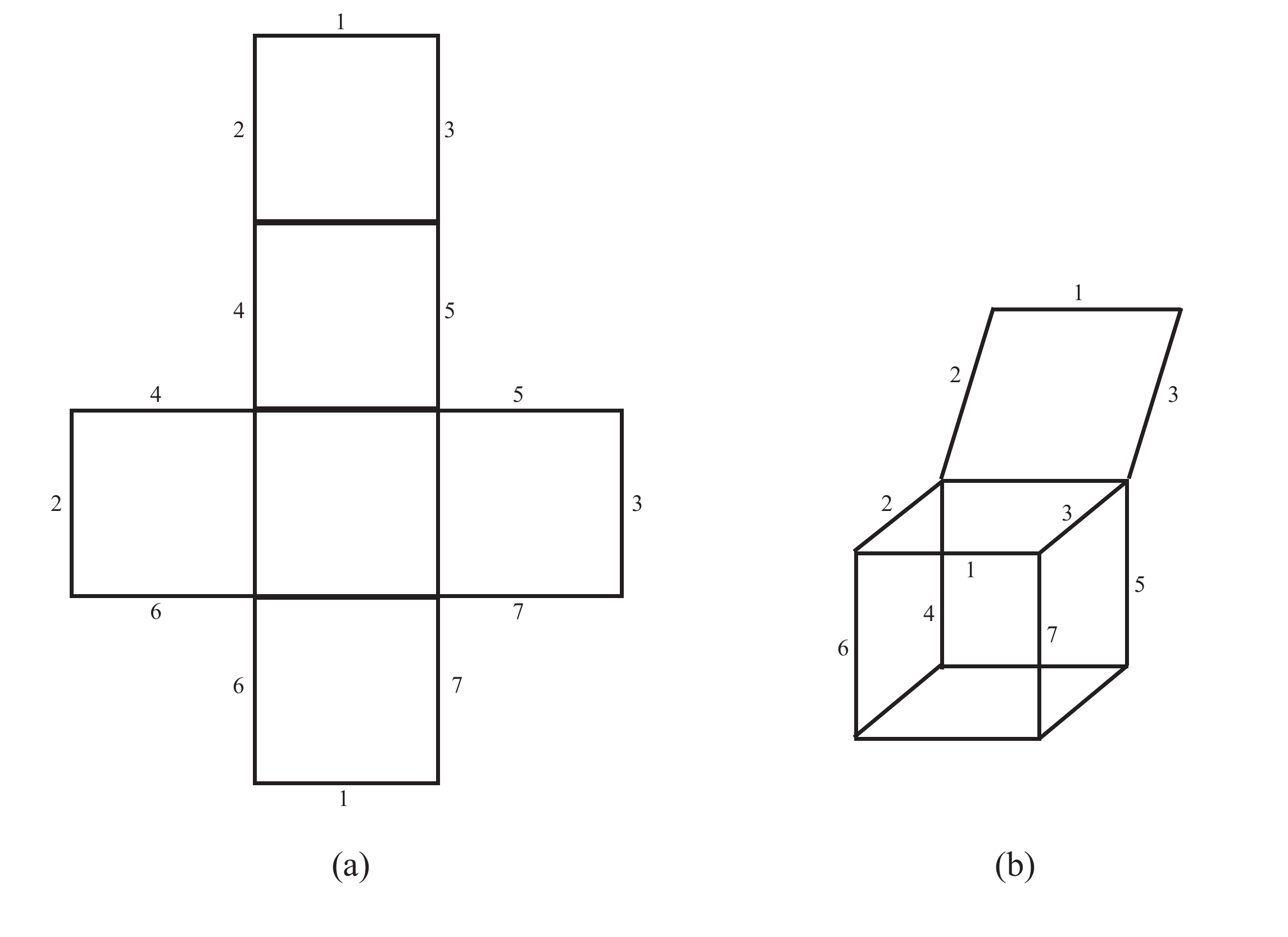}
\caption{Schematic for construction of a three-dimensional cube from two-dimensional tiles.  This figure follows Figure 1 in~\cite{3D}.  Tiles floating in solution combine in a group of six, as in (a), which then folds into a box, as in (b).  Each face of the cube can then bind to other cube faces, because of additional glues encoded into the faces that, for clarity, are not shown in the diagram.} \label{figure:boxfold}
\end{figure}
\begin{theorem} \label{theorem:possibility}
Let $\mathcal{M}$ be a system of $n$ distributed processes ($n \geq 3$) and one $n$-consensus object $O$, such that the processes do not send one another messages, each process invokes $O$ at most once, and all correctly-working processes will run forever.  There is a three-dimensional tile assembly system that simulates solution of the Consensus Subroutine Problem for $\mathcal{M}$.
\end{theorem}
\begin{proof}
Let $\mathcal{M}$ be a distributed system as in the theorem statement.  We will simulate the execution of $n$ processes in two tiling dimensions by constructing a tile system that assembles in a ``starfish'' pattern, with $n$ arms growing out from a central seed tile.  Each arm simulates one process of $\mathcal{M}$.  We can control the direction of growth---and the rate of increase in width---of each tile configuration that simulates a Turing machine, so such starfish configurations can be tiled, for any $n$.  (Techniques to do this are discussed in~\cite{ccsa}~\cite{patsum}.)
The construction of Figure~\ref{figure:simulation} was successful because the communication between $\pi_1$ and $\rho$, and $\pi_2$ and $\rho$, did not block the growth of either $\pi_1$ or $\pi_2$.  We duplicate that effect by placing the tile configuration that simulates $n$-consensus object $O$ \emph{above} the seed assembly (in the $z$-direction), and requiring that the tiles that simulate processes have glue names on their ``tops'' (their upward faces in the $z$-direction) that create paths for transmission of proposed decision values, and reception of acks of those transmissions.  Figure~\ref{figure:four-consensus-object-2} explicitly shows how to construct a three-dimensional tile configuration that simulates a 4-consensus object.  We could simulate a 3-consensus object by blocking one of the sides of the ``approach pad'' of that construction.  We now argue generally about how to construct such a simulating configuration for any $n \in \mathbb{N}$.

The subassembly in Figure~\ref{figure:four-consensus-object-2} sends a ray due north, and receives an ack due south.  If we want to simulate an $n$-consensus object for $n>4$, we use the following more general tiling method: each of the $n$ processes has an entry square to the approach pad.  Once a transmission ray is incident to the approach pad, its tiles bind in a way that ``floods'' the approach pad---greedily binding to any location possible on the approach pad.  Since we are assuming there can be no tile blockages on the approach pad, we are guaranteed that \emph{some} tile from \emph{some} unblocked $\pi_i$ will bind to the decision point.  Once a tile binds to the decision point, acks are built in the $z$-direction above the decision point, and the correctly-labeled ack radiates along the tops of all the tiles that bound to the approach pad, and follow back the transmission rays to the $\pi_j$'s that sent them.  Because we have the extra $z$-direction to work in, it is ok if the transmission rays on the approach pad block one another: each $\pi_i$ will attempt to enter the approach pad at a unique location, and if its entry (or forward progress on the approach pad) is blocked, its presence nonetheless creates a ``channel'' for the ack information to be transferred, along the top and the sides of the transmission ray.
\end{proof}
\begin{figure}
\centering
\includegraphics[height=5.5in]{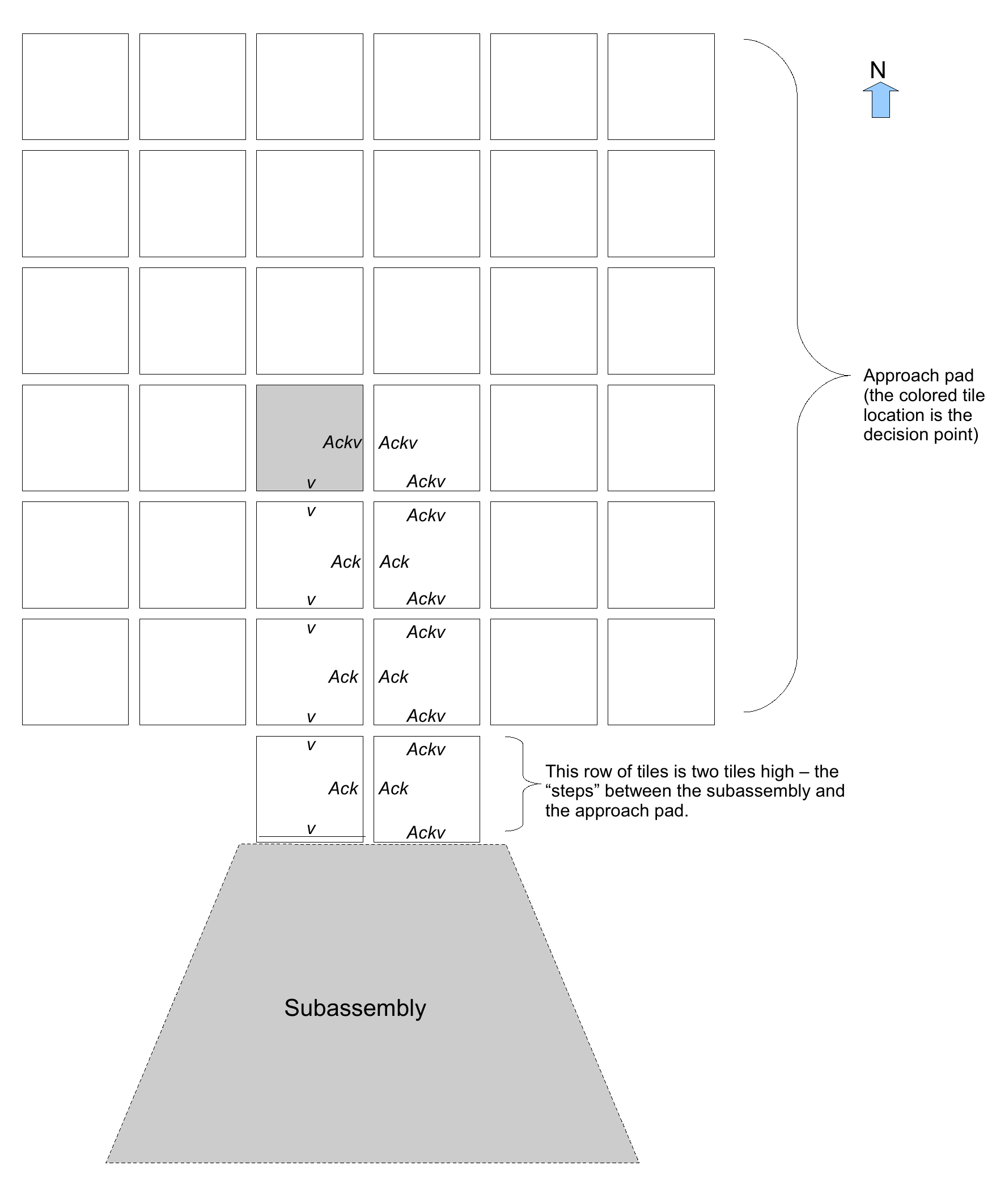}
\caption{Proposed values from subassemblies can approach the decision point from four sides.  Whichever value binds to the decision point is the one chosen and acked to all subassemblies.  This scales up to any finite $n \geq 3$, by making the approach pad large enough to accomodate ``transmission rays'' from $n$ subassemblies; and by sending the acks along the ``tops'' of the ``transmission rays,'' as shown in the side view, Figure~\ref{figure:four-consensus-object}.  (That way, each subassembly will receive an ack, and not be blocked out from lack of space.)  Note this construction is not ``fair,'' in the sense that the decision point is not equidistant from each subassembly, so one subassembly's proposed value may be more likely to be chosen than that of others.  Fortunately, consensus only requires common agreement, not fairness.} \label{figure:four-consensus-object-2}
\end{figure}
\begin{figure}
\centering
\includegraphics[height=5in]{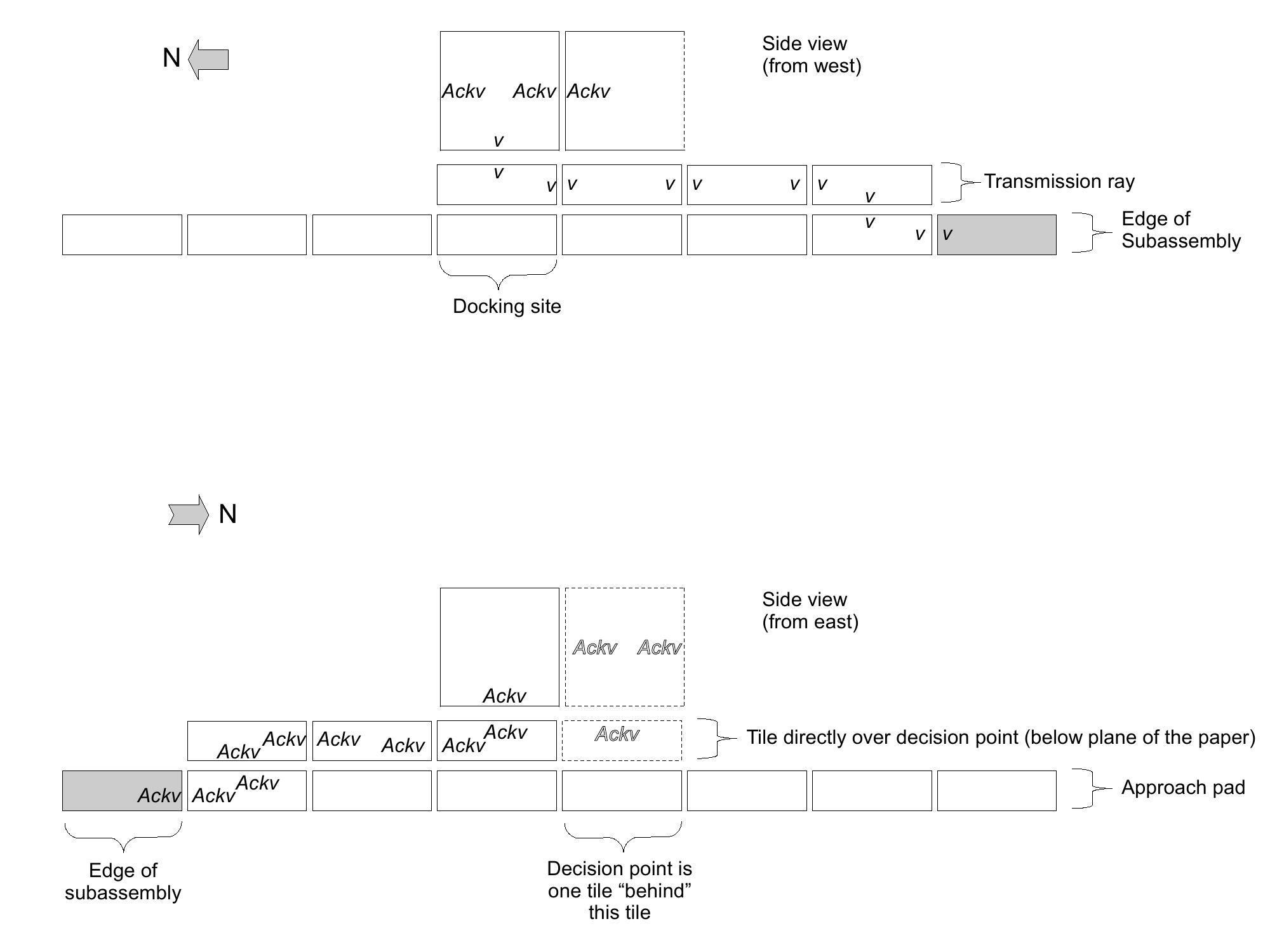}
\caption{Side views of the tile configuration that can simulate an $n$-consensus object for finite $n \geq 3$.  Each subassembly that is simulating a process sends a transmission ray along the approach pad.  The first process that reaches the decision point receives an ack of its own value.  Other transmission rays, once the bind adjacent to a location that has an ``ack cube'' already placed in the $z$-direction, will receive an ack of the already-decided value.  The ack returns to the subassembly along the ``top'' of the transmission ray.} \label{figure:four-consensus-object}
\end{figure}
\section{Conclusion} \label{section:conclusion}
We have shown how two-dimensional tile assembly systems can simulate solution to the consensus problem for some two-process distributed systems, and how three-dimensional tile assembly systems can simulate a strengthening of the consensus problem for some $n$-process distributed systems, for any $n$.  One way to extend our current results would be to consider what types of communication among processes could be simulated by tile assembly.

The error assumption in this paper---that any binding fault causes an entire subassembly to stop growing permanently---is an oversimplification that renders sophisticated error-management techniques unavailable for almost trivial reasons.  One can imagine a simulation of a two-processor system in which the tile constructions that simulate each processor are interleaved somehow, with sufficient redundancy that many binding errors could be recovered, in a way analogous to the behavior of resilient cellular automata~\cite{gacs}.  However, because those assemblies are interleaved, the assumption of blockage failure would cause them both to stop growing, as their frontiers are not independent of each other.  An important area of future work will be the modeling of more realistic failures, and the development of techniques to manage such failures.
\section*{Acknowledgements}
I am grateful to James Aspnes, Jim Lathrop, Jack Lutz and Scott Summers for helpful discussions on earlier versions of this paper.  I am especially grateful to Soma Chaudhuri for many helpful discussions.
\end{document}